\documentclass[preprint,12pt]{temp}

\usepackage{amssymb}
\usepackage{lineno}

\usepackage{amssymb}

\usepackage{amsthm}

\usepackage{amsmath}

\usepackage{multirow}

\usepackage{caption}
\usepackage{subcaption}

\usepackage{tikz}

\usetikzlibrary{decorations.pathreplacing,calligraphy}

\newtheorem{theorem}{Theorem}
\newtheorem{lemma}[theorem]{Lemma}
\newtheorem{definition}[theorem]{Definition}

\newtheorem{observation}[theorem]{Observation}
\newtheorem{corollary}[theorem]{Corrolary}
\newtheorem{algo}{Algorithm}

\journal{}

\begin{document}

\begin{frontmatter}

%% Title, authors and addresses

%% use the tnoteref command within \title for footnotes;
%% use the tnotetext command for theassociated footnote;
%% use the fnref command within \author or \affiliation for footnotes;
%% use the fntext command for theassociated footnote;
%% use the corref command within \author for corresponding author footnotes;
%% use the cortext command for theassociated footnote;
%% use the ead command for the email address,
%% and the form \ead[url] for the home page:
%% \title{Title\tnoteref{label1}}
%% \tnotetext[label1]{}
%% \author{Name\corref{cor1}\fnref{label2}}
%% \ead{email address}
%% \ead[url]{home page}
%% \fntext[label2]{}
%% \cortext[cor1]{}
%% \affiliation{organization={},
%%            addressline={}, 
%%            city={},
%%            postcode={}, 
%%            state={},
%%            country={}}
%% \fntext[label3]{}

\title{New Bounds for Time-Dependent Scheduling with Uniform Deterioration}

\author{Angelos Gkikas\fnref{kings,ocado}}
\author{Dimitrios Letsios\fnref{kings}}
\author{Tomasz Radzik\fnref{kings}}
\author{Kathleen Steinh\"ofel\fnref{kings}}

% %% use optional labels to link authors explicitly to addresses:
% %% \author[label1,label2]{}
% \affiliation[kings]{organization={King's College London},
% %             addressline={},
% %             city={},
% %             postcode={},
% %             state={},
%             country={United Kingdom}}

% \affiliation[ocado]{organization={Ocado Technology},
% %             addressline={},
%             city={Hatfield},
%             postcode={AL10 9UL},
% %             state={},
%             country={United Kingdom}}

\address[kings]{King's College London, United Kingdom}
\address[ocado]{Ocado Technology, Hatfield, AL10 9UL, United Kingdom}

\begin{abstract}
Time-dependent scheduling with linear deterioration involves determining when to execute jobs whose processing times degrade as their beginning is delayed.
Each job $i$ is associated with a release time $r_i$ and a processing time function $p_i(s_i)=\alpha_i + \beta_i\cdot s_i$, where $\alpha_i,\beta_i>0$ are constants and $s_i$ is the job's start time.
In this setting, the approximability of both single-machine minimum makespan and total completion time problems remains open.
Here, we take a step forward by developing new bounds and approximation results for the interesting special case of the problems with uniform deterioration, i.e.\ $\beta_i=\beta$, for each $i$.
The key contribution is a $O(1+1/\beta)$-approximation algorithm for the makespan problem and a $O(1+1/\beta^2)$ approximation algorithm for the total completion time problem.
Further, we propose greedy constant-factor approximation algorithms for instances with $\beta=O(1/n)$ and $\beta=\Omega(n)$, where $n$ is the number of jobs.
Our analysis is based on a new approach for comparing computed and optimal schedules via bounding pseudomatchings.

% The approximability of the time-dependent scheduling problem remains unsettled.
% This manuscript takes a step forward by considering uniform deterioration.

% Time-dependent scheduling requires solving temporal resource allocation problems with job processing times affected by the start times.
% Such problems are relevant to various application domains, including transportation, manufacturing and military defence, and have been the topic of fruitful investigations over the last 40 years. Our study is motivated by a generalisation of the problem of intercepting escaping targets \cite{Helvig1998}.
% We focus on the fundamental makespan minimization problem with a set of jobs that arrive over time and linearly increasing processing times.
% Assuming that each job $i$ is associated with a release time $r_i$ and a processing time function $p_i(s_i)=\alpha_i+\beta s_i$ of the job's start time $s_i$, the objective is to schedule the jobs on a single machine so that the makespan, i.e.\ the time at which the last job completes, is minimized.
% The problem is known to be strongly $\mathcal{NP}$-hard but its approximability remains open.
% We propose a $O(1+\frac{1}{\beta})$-approximation algorithm as well as constant-factor approximation algorithms for the special cases $\beta\leq 1/n$ and $\beta\geq n+1$. 
% The proposed algorithms are semi-online and allow coping with the lack of online algorithms with a constant competitive ratio for variants of the problem. 
% We extend obtained results to minimizing the sum of completion times.
\end{abstract}

% %%Graphical abstract
% \begin{graphicalabstract}
% %\includegraphics{grabs}
% \end{graphicalabstract}

%%Research highlights
% \begin{highlights}
% \item We study single-machine machine scheduling problems with uniformly deteriorating jobs and develop greedy heuristics whose performance is substantiated with provable guarantees.
% \item We propose a new approach for analyzing computed solutions to the problems of interest via bounding pseudomatchings.
% \item We obtain an $O(1+1/\beta)$-approximation algorithm for the makespan problem. 
% \item We demonstrate relationships between the makespan and total completion time objectives, and derive an $O(1+1/\beta^2)$-approximation algorithm for the latter.
% % including algorithms achieving $O(1+1/\beta)$ and $O(1+1/\beta^2)$ approximation ratios for minimizing the makespan and the total completion time, respectively. 
% % \item We show the existence of constant-factor approximation algorithms for the makespan and sum of completion times problems.
% \end{highlights}

\begin{keyword}
Time-Dependent Scheduling \sep Linear Deterioration \sep Approximation Algorithms
%% keywords here, in the form: keyword \sep keyword
%% PACS codes here, in the form: \PACS code \sep code
%% MSC codes here, in the form: \MSC code \sep code
%% or \MSC[2008] code \sep code (2000 is the default)
\end{keyword}

\end{frontmatter}

% \newpage

% \linenumbers

%% main text
\section{Introduction}
\label{}

Single-machine scheduling problems involve deciding when to process a set $\mathcal{J}=\{1,\ldots,n\}$ of $n$ jobs arriving over time, i.e.\ each job $i\in\mathcal{J}$ is associated with a release time $r_i\in\mathbb{R}^+$, using a machine that may execute at most one job per time so as to optimize some objective function $f(\vec{C})$, e.g.\ the makespan $\max_{i\in\mathcal{J}}\{C_i\}$ or the total completion time $\sum_{i\in\mathcal{J}}C_i$, where $\vec{C}=(C_1,\ldots,C_n)$ is the vector of job completion times. 
Prior literature largely assumes that each job $i\in\mathcal{J}$ has a fixed processing time $p_i\in\mathbb{R}^+$. 
However, this assumption can be fairly strong.
In various contexts, e.g.\ production scheduling with machine degradation and delivery scheduling in road networks with varying traffic, the time at which the execution of a job begins significantly affects its processing time.
Scheduling problems where the processing time of each job $i\in\mathcal{J}$ is a function $p_i(s_i)$ of its start time $s_i$ are typically referred to as \emph{time-dependent scheduling problems}.

Previous work investigates time-dependent scheduling problems with processing time functions $p_i(s_i)=\alpha_i+\beta_i s_i$, where $\alpha_i\in\mathbb{R}^+$ is the \emph{fixed part} and $\beta_i\cdot s_i$ is the \emph{variable part} depending on the deterioration rate $\beta_i\in\mathbb{R}^+$ and the start time $s_i$, for each job $i\in\mathcal{J}$. 
Such problems are referred to as \emph{scheduling with linear deterioration} and model settings where delaying the beginning of a job execution by one unit of time results in an increase of the job's processing time by $\beta_i$ units of time.  
In this context, the single-machine problems of minimizing the makespan and the total completion time are open from an algorithmic viewpoint. 
Using the standard 3-field scheduling notation, the problems can be denoted as $1|r_i,p_i(t)=\alpha_i+\beta_i\cdot s_i|C_{\max}$ and $1|r_i, p_i(t)=\alpha_i+\beta_i\cdot s_i|\sum C_i$.
When all jobs have equal release times, the makespan problem is polynomially solvable \cite{Gupta1988,Mosheiov1994}, while the complexity of the total completion time problem is unknown and conjectured to be $\mathcal{NP}$-hard \cite{Gawiejnowicz2020}.
When the jobs have arbitrary release times, both problems are known to be strongly $\mathcal{NP}$-hard \cite{Cheng1998,Cheng1998b,Lenstra1977}.
% the makespan problem is $\mathcal{NP}$-hard even for instances with two distinct release times and uniform deterioration rates, i.e.\ $\beta_i=\beta$, for each $i\in\mathcal{J}$ \citep{Cheng1998}, while the total completion time problem is 
% and their approximability remains unsettled.
The best known algorithms are based on iterative subproblem decomposition, e.g.\ dynamic programming and branch-and-bound \cite{Bosio2009,DeSouza2022,Lee2008}, but have exponential running times.

The two problems have attracted attention in the special cases with
% On the positive side, there is better understanding of approximation bounds for the special cases of the problems with 
(1) proportional linear deterioration, i.e.\ $p_i(t)=\beta_i\cdot s_i$ (eq.\ $\alpha_i=0$), for each $i\in\mathcal{J}$, and (2) fixed processing times, i.e.\ $p_i(s_i)=\alpha_i$ (eq.\ $\beta_i=0$), for each $i\in\mathcal{J}$. 
%(see Table~\ref{Table:Approximation_Bounds}).
In the former case, the makespan problem $1|r_i,p_i(s_i)=\beta_i\cdot t|C_{\max}$ is optimally solvable in $O(n\log n)$ time \cite{Liu2012,Miao2012}, while the best known algorithm for the total completion problem $1|r_i,p_i(s_i)=\beta_i\cdot t|\sum C_i$ is $(1+\beta_{\max})$-approximate \cite{Chai2020}.
In the latter case, makespan problem is polynomially solvable via a greedy algorithm \cite{Lawler1973}, while the total completion time problem is strongly $\mathcal{NP}$-hard, admitting a Polynomial-Time Approximation Scheme (PTAS) and greedy constant-factor approximation algorithms \cite{Afrati1999,Karger2010}.

In addition to the above, there exist various complexity and approximation results for problem generalizations (e.g.\ multiprocessor environments \cite{Chen1996,Ji2009,Kang2007,Li2014,Liu2013,Yu2013}), relaxations (e.g.\ preemptive versions \cite{Ng2010}), and variants (e.g.\ step and position-dependent processing time functions \cite{Cheng2001,Yang2021} and uncertainty \cite{Letsios2021royal}).
Surveys relevant to time-dependent scheduling algorithms can be found in \citep{Cheng2004,Gawiejnowicz2020,Letsios2020}.
Table~1 summarizes results closely related to this manuscript.
Last but not least, there exist investigations on interesting time-dependent scheduling applications, including production and delivery scheduling \citep{He2020,Letsios2021lex}, defending aerial threats \cite{Helvig1998}, fire fighting \citep{Pappis2010}, and personel scheduling \citep{Xu2021}.

\begin{table}
% \begin{center}
\begin{tabular}{ |lccc| } 
\hline
\multirow{2}{*}{Scheduling Problem} & \multirow{2}{*}{Complexity} & Best Known & 
\multirow{2}{*}{Ref.} \\
& & Algorithm & \\
\hline
\multicolumn{4}{|l|}{\textbf{Linear Deterioration}} \\
$1|p_i(t)=\alpha_i+\beta_i\cdot t|C_{\max}$ & $\mathcal{P}$ & $O(n\log n)$-time & \cite{Gupta1988} \\
$1|r_i,p_i(t)=\alpha_i+\beta_i\cdot t|C_{\max}$ & $\mathcal{NP}$-hard & Exp.\ & \cite{Cheng1998} \\
$1|p_i(t)=\alpha_i+\beta_i\cdot t|\sum C_i$ & ? & Exp.\ & \cite{Gawiejnowicz2020} \\
$1|r_i,p_i(t)=\alpha_i+\beta_i\cdot t|\sum C_i$ & $\mathcal{NP}$-hard & Exp.\ & \cite{Lenstra1977} \\
\hline
\multicolumn{4}{|l|}{\textbf{Uniform Deterioration}} \\
$1|r_i,p_i(t)=\alpha_i+\beta\cdot s_i|C_{\max}$ & $\mathcal{NP}$-hard & $(1+1/\beta)$-approx.\ & [*] \\
$1|r_i,p_i(t)=\alpha_i+\beta\cdot s_i|\sum C_i$ & $\mathcal{NP}$-hard & $(1+1/\beta^2)$-approx.\ & [*] \\
\hline
\multicolumn{4}{|l|}{\textbf{Simple Deterioration}} \\
% $1|p_j(t)=\beta_j\cdot t|C_{\max}$ & $O(n)$-time & \citep{Mosheiov1994} \\
% $1|p_j(t)=\beta_j\cdot t|\sum C_j$ & $O(n\log n)$-time & \citep{Mosheiov1994} \\
$1|r_i,p_i(t)=\beta_i\cdot t,pmtn|C_{\max}$ & $\mathcal{P}$ & $O(n\log n)$-time & \citep{Ng2010} \\
$1|r_i,p_i(t)=\beta_i\cdot t|C_{\max}$ & $\mathcal{P}$ & $O(n\log n)$-time & \citep{Liu2012,Miao2012} \\
$1|r_i,p_i(t)=\beta_i\cdot t,pmtn|\sum C_i$ & $\mathcal{P}$ & $O(n\log n)$-time & \citep{Ng2010} \\
% $1|r_j,p_j(t)=\beta_j\cdot t|C_{\max}$ & 2-approximation & \citep{Miao2012} \\
$1|r_i,p_i(t)=\beta_i\cdot t|\sum C_i$ & $\mathcal{NP}$-hard & $(1+\beta_{\max})$-approx & \citep{Liu2012} \\
\hline
\multicolumn{4}{|l|}{\textbf{Fixed Processing Times}} \\
$1|r_i,p_i(t)=\alpha_i|C_{\max}$ & $\mathcal{P}$ & $O(n\log n)$-time & \citep{Lawler1973} \\ 
$1|r_i,p_i(t)=\alpha_i|\sum C_j$ & $\mathcal{NP}$-hard & PTAS & 
\citep{Lenstra1977,Afrati1999} \\ 
\hline
\end{tabular}
% \end{center}
\label{Table:Approximation_Results}
\caption{Algorithmic results for single-machine time-dependent scheduling problems. The stars [*] indicate results obtained in the current manuscript.}
\end{table}

\paragraph{Contributions and paper organization}

Despite the aforementioned literature, the approximability of the single-machine time-dependent scheduling problems with jobs arriving over time, linear deterioration, the makespan and total completion time objectives remains unsettled.
This manuscript takes a step forward in this direction by focussing on the special case with uniform deterioration, i.e.\ the problems $1|r_i,p_i(s_i)=\alpha_i+\beta\cdot C_{\max}$ and $1|r_i,p_i(s_i)=\alpha_i+\beta\cdot \sum C_i$.
To the authors knowledge, no approximation algorithms are known for those. 
%that models scheduling with machine degradation, which is particularly useful for maintenance operations.
Our main contribution is the analysis of greedy algorithms and the derivation of approximation results based on a new approach for bounding time-dependent scheduling problems using pseudomatchings.
The first part of the manuscript (Sections 2-5) is devoted to the makespan problem and the last part (Section 6) covers the total completion time problem.
In more detail, the manuscript proceeds as follows.

% We take a step forward and elaborates on single-machine time-dependent scheduling problems with uniform processing time functions $p_i(s_i)=\alpha_i+\beta\cdot s_i$ (eq.\ $\beta_i=\beta$), for each $i\in\mathcal{J}$. 
% Despite the aforementioned literature, the approximability of the single-machine time-depedent scheduling problems with the makespan and total completion time objectives remains unsettled.
% This manuscript makes a step forward with new algorithmic approaches for the special case of the problem with uniform deterioration.
% This setting is particularly useful modeling settings with machine degradation.
% We show that existing greedy algorithms do not attain good approximation ratio and derive positive and negative results.
% Our main approximation results rely on a novel way for obtaining bounds to time-dependent scheduling problems that we call sum-bounding pseudo-matchings.

Section~\ref{Section:Preliminaries} formally describes the problem and expresses the makespan of a feasible schedule as a function of the job processing times and idle periods. 
Section~\ref{Section:Pseudomatchings} introduces our pseudomatching concepts and demonstrates their bounding properties.
Section~\ref{Section:Basic_Algorithms} analyzes two basic algorithms that we call \emph{Non-Idling} and \emph{Non-Interfering} \cite{Chretienne2016,Karger2010}.
The former avoids idle periods by always scheduling a pending job when the machine becomes available, while the latter introduces idle periods (i.e.\ delays pending jobs) to prioritize (i.e.\ avoid interfering with) jobs that arrive later\footnote{A job is \emph{pending} if it has been released, but has not been processed.}. 
Non-Idling optimally solves the problems $1|r_j,p_j(s_j)=\beta_j\cdot s_j|C_{\max}$ and $1|r_j,p_j(s_j)=\alpha_j|C_{\max}$ and is similar to other standard scheduling algorithms, e.g.\ List Scheduling and First-Come First Served. 
Non-Interfering extends the standard Shortest Processing Time First algorithm, which is optimal for $1|p_j(s_j)=\alpha_j+\beta\cdot t|C_{\max}$ instances with a single release time, to instances with arbitrary release times. 
On the positive side, we prove that the two algorithms achieve constant approximation ratios for the special cases of our problem with $\beta\leq 1/n$ and $\beta\geq n+1$, respectively.
On the negative side, we show that both algorithms are $\Omega((1+\beta)^n)$-approximate in the worst case.

The above negative result demonstrates an interesting degeneracy of the problem: misplacing just a single job may result in severe solution quality degradation.
Despite this pathological finding, Section~\ref{Section:ECTF} shows the existence of a time-dependent priority policy, namely Earliest Completion Time First, that achieves an $(3+1/\beta)$-approximation ratio for $1|r_j,p_j(s_j)=\alpha_j+\beta\cdot t|C_{\max}$.  
Next, we turn our attention to the total completion time objective.
Prior literature derives equivalence relationships between the makespan and the sum of completion times.
For single-machine instances with fixed processing times, an optimal schedule for makespan is 2-approximate for the total completion time and vice versa \cite{Stein1997}.
We extend this approximation equivalence relationship to the time-dependent scheduling context.
On one hand, we show that any $\rho$-approximation algorithm for the total completion time problem is $(1+\rho)$-approximate for minimizing makespan.
On the other hand, we show that $\rho$-approximation algorithm for makespan is $(1+1/\beta)\rho$-approximate for the total completion time. 
This last finding implies the existence of a $O(1+1/\beta^2)$-approximation algorithm for the total completion time problem.

\section{Preliminaries}
\label{Section:Preliminaries}

Next, we formally define the time-dependent scheduling problem with uniformly deteriorating processing times and express the makespan of a feasible schedule as a weighted sum of the fixed processing times and gap lengths. %, show that the Shortest-Fixed Processing Time First algorithm achieves a poor approximation ratio for the problem. 
%and define the notion of a dominating matching which allows analyzing algorithms with better approximation ratios. 

% \paragraph{Problem Definition.}
\subsection{Problem Definition.}
An instance of the problem consists of a set $\mathcal{J}=\{1,\ldots,n\}$ of jobs that have to be executed by a single machine that may execute at most one job at a time.
Each job must be executed non-preemptively, i.e.\ in a continuous time interval without interruptions until it completes.
For each job $i\in\mathcal{J}$, a linearly increasing function $p_i(s_i)=\alpha_i+\beta\cdot s_i$  
specifies the processing time of $i$ if it begins at time $s_i$.
We refer to the terms $\alpha_i$ and $\beta\cdot s_i$ as \emph{fixed} and \emph{variable} processing time, respectively, where $\beta>0$ is a constant rate at which $p_i(s_i)$ is increased per unit of time that the start time $s_i$ is delayed.
Given two jobs $i,j\in\mathcal{J}$, if $\alpha_i<\alpha_j$, then we say that $i$ is \emph{shorter} than $j$ and that $j$ is \emph{longer} than $i$. 
Job $i\in\mathcal{J}$ is released at time $r_i$, i.e.\ may only begin processing at a time $s_i\geq r_i$.
W.l.o.g.\ $r_{\min}=\min_{i\in\mathcal{J}}\{r_i\}=0$.
Let $C_i$ be the completion time of job $i$, i.e.\ $C_i=\alpha_i+(1+\beta)s_i$. 
The objective is to find a feasible schedule such that the makespan, i.e.\ the time $T=\max_{i\in\mathcal{J}}\{C_i\}$ at which the last job completes, is minimized.
Given a time $t$, denote by $\mathcal{P}(t)$ the set of \emph{pending jobs}, i.e.\ the ones which have been released but have not begun processing before $t$. 
At each time $t$ that the machine becomes available, a feasible schedule specifies the next job to begin from time $t$ and onward.

% \paragraph{Makespan Expression.}
\subsection{Makespan Expression.}
Due to release times, optimal schedules may require \emph{gaps}, i.e.\ maximal idle time intervals during which no job is processed (Figure~\ref{Figure:Compact_Schedule}).
Consider a feasible schedule $\mathcal{S}$ and number the jobs in increasing order $s_1<\ldots<s_n$ of their start times. 
Denote the gap between jobs $i-1$ and $i$ by $q_i=s_i-C_{i-1}$, for $i\in\{1,\ldots,n\}$, where $C_0=0$. 
If $q_i=0$, then there is no idle period between jobs $i-1$ and $i$.
Lemma~\ref{Lemma:Delay_Propagation} expresses the makespan of a feasible schedule w.r.t.\ gaps and fixed processing times.
This is an adaptation of standard expressions in the time-dependent scheduling literature \cite{Gawiejnowicz2020}, but now accounts for gaps because release times.
Lemma~\ref{Lemma:New_Makespan_Expression} derives an alternative expression of the fixed processing time contributions to the makespan.

\begin{figure}[!h]
\begin{center}
\includegraphics[scale=0.55]{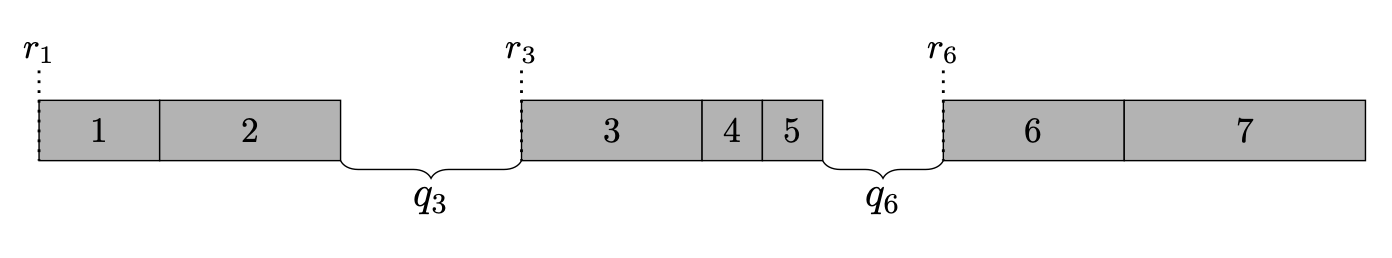}
\end{center}
\caption{Illustration of feasible schedule with seven jobs and two gaps, e.g.\ during the time interval $[C_2,s_3)$. There is an optimal schedule such that, if job $i$ begins right after a gap, i.e.\ $q_i>0$, then $s_i=r_i$.}
\label{Figure:Compact_Schedule}
\end{figure}

\begin{lemma}
\label{Lemma:Delay_Propagation}
Consider a feasible schedule $\mathcal{S}$ and suppose that the jobs are numbered in increasing order of their start times in $\mathcal{S}$, i.e.\ $s_1\leq\ldots\leq s_n$. 
Then, the makespan of $\mathcal{S}$ is:
\begin{align*} 
T=\sum_{i=1}^n(1+\beta)^{n-i+1}q_i+\sum_{i=1}^n(1+\beta)^{n-i}\alpha_i
\end{align*}
\end{lemma}
\begin{proof}
We show by induction on $k\in\{1,\ldots,n\}$ that $C_k=\sum_{i=1}^k(1+\beta)^{k-i}[(1+\beta)q_i+\alpha_i]$.
For the induction basis, it clearly holds that $C_1=(1+\beta)q_1+\alpha_1$, since job $1$ begins at time $s_1=q_1$.
For the induction step, suppose that the equality is true with index $k-1$.
Using the fact that $s_k=C_{k-1}+q_k$ and the induction hypothesis:
\begin{align*}
C_k & = (1+\beta)[C_{k-1}+q_k] + \alpha_k  \\
& = (1+\beta)\left[\sum_{i=1}^{k-1}(1+\beta)^{(k-1)-i}
\left[(1+\beta)q_i+\alpha_i\right] + q_k\right] + \alpha_k \\
& = \sum_{i=1}^{k}(1+\beta)^{k-i}\left[(1+\beta)q_i+ \alpha_i\right]
\end{align*}
\end{proof}

Lemma~\ref{Lemma:Delay_Propagation} has the following implications. 
First, if all jobs begin at time $t$ and are executed without any gap between them, then they complete at $T = (1+\beta)^nt+\sum_{i=1}^n(1+\beta)^{n-i}\alpha_i$. 
Second, when all jobs have equal release times, there exists always an optimal schedule without gaps and greedily scheduling the jobs in non-decreasing order $\alpha_1\leq\ldots\leq\alpha_n$ of their fixed processing times is optimal \cite{Gawiejnowicz2020}.
Third, for any subset $\mathcal{J}'=\{\gamma(1),\ldots,\gamma(k)\}$ of jobs sorted in non-decreasing order $\alpha_{\gamma(1)}\leq\ldots\leq\alpha_{\gamma(k)}$ of their fixed processing times and executed consecutively without gaps starting at time $t$, we get the lower bound $T\geq (1+\beta)^kt+ \sum_{i=1}^k(1+\beta)^{k-i}\alpha_{\gamma(i)}$ on the makespan of any feasible schedule.

% Lemma~\ref{Lemma:New_Makespan_Expression} derives an alternative expression of the fixed processing time contributions to the makespan.
%  in an arbitrary feasible schedule for the problem.
% Table~\ref{Table:Processing_Times} shows the processing time of each job according to this expression, in a feasible schedule without gaps.

\begin{lemma}
\label{Lemma:New_Makespan_Expression}
Consider a feasible schedule $\mathcal{S}$ and number the jobs in increasing order $s_1\leq\ldots\leq s_n$ of their start times in $\mathcal{S}$.
Then $S$ has fixed processing time cost $\sum_{i=1}^n(1+\beta)^{n-i}\alpha_i=\sum_{i=1}^n\alpha_i+\sum_{k=2}^{n}\beta(1+\beta)^{n-k}\left(\sum_{i=1}^{k-1}\alpha_i\right)$.
% Then, $\mathcal{S}$ has makespan $T=\sum_{i=1}^n\alpha_i+\sum_{k=2}^{n}\beta(1+\beta)^{n-k}\left(\sum_{i=1}^{k-1}\alpha_i\right)$.
\end{lemma}
\begin{proof}
Consider a job $i\in\mathcal{J}$.
Because $i$ is scheduled in the $i$-th position of $\mathcal{S}$, its execution increases the start time, and therefore the processing time, of every job in the set $\{i+1,,\ldots,n\}$.
Specifically, the processing time of job $i+1$ is increased by $\beta\alpha_i$, of job $i+2$ by $\beta(1+\beta)\alpha_i$ and so on. 
That is, the processing time of job $n$ is increased by $\beta(1+\beta)^{n-i-1}\alpha_i$.
Hence, the overall contribution of $\alpha_i$ to the makespan is $1+\sum_{j=1}^{n-i}\beta(1+\beta)^{j-1}\alpha_{i}=(1+\beta)^{n-i}\alpha_i$.
Using this geometric series sum and Lemma~\ref{Lemma:Delay_Propagation}, the fixed processing time distribution to $\mathcal{S}$ can be expressed $A=\sum_{i=1}^n\alpha_i+\sum_{k=1}^{n-1}\left(\sum_{i=1}^{n-k}\beta(1+\beta)^{n-k-i}\right)\alpha_k$.
Next, we rearrange the sum so that fixed processing time terms $\alpha_i$ with the same weight $\beta(1+\beta)^{n-k}$ are grouped together, for $i\in\mathcal{J}$ and $k\in\{2,\ldots,n\}$.
That is, $A=\sum_{i=1}^n\alpha_i+\sum_{k=2}^{n}\beta(1+\beta)^{n-k}\left(\sum_{i=1}^{k-1}\alpha_i\right)$.
\end{proof}

% \begin{table}[!h]
%     \centering
%     \begin{tabular}{|c|c|}
%         \hline
%         \textbf{Job} & \textbf{Processing Time}  \\
%          \hline
%          $1$ & $\alpha_1$ \\
%          $2$ & $\beta\alpha_1+\alpha_2$ \\
%          $3$ & $\beta(1+\beta)\alpha_1+\beta\alpha_2+\alpha_3$ \\
%          \vdots & \vdots \\
%          $n$ & $\beta(1+\beta)^{n-2}\alpha_1 + \beta(1+\beta)^{n-3}\alpha_2 + \ldots + \beta(1+\beta)\alpha_{n-2} + \beta\alpha_{n-1}+\alpha_n$ \\
%          \hline
%     \end{tabular}
%     \caption{The job processing times in a schedule $\mathcal{S}$ without gaps, i.e.\ $q_i=0$, for each $i\in\mathcal{J}$, where the jobs are numbered in increasing order $s_1<\ldots<s_n$ of their start times.}
%     \label{Table:Processing_Times}
% \end{table}

\section{Bounding Pseudomatchings}
\label{Section:Pseudomatchings}

To analyze the performance of our algorithms for $1|r_j,p_j(t)=\alpha_j+\beta\cdot t|C_{\max}$, we need an approach for upper and lower bounding the (fixed processing time) load completed by a feasible and an optimal schedule, respectively, up to any time $t$.
To this end, we introduce the \emph{$\rho$-pseudomatching} and \emph{weak pseudomatching} concepts that allow bounding sums and geometric series incorporating the $\beta$ parameter, respectively.
% The former is used for analyzing the Non-Interfering algorithm (Section~4), while the latter is used for analyzing the Earliest Completion Time First Algorithm (Section~5).
Definition~\ref{Def:Bounding_Graph} defines the so-called \emph{bounding graph} that is used for comparing schedules computed by algorithms with optimal schedules.
Definition~\ref{Def:S_Pseudomatchings} and Lemma~\ref{Lem:S_Pseudomachings} summarize a core argument used for analyzing the Non-Interfering algorithm (Section~4).
Definition~\ref{Def:G_Pseudomatchings} and Lemma~\ref{Lem:G_Pseudomachings} describe a main argument in the analysis of the Earliest Completion Time First algorithm (Section~5).
The main technical difficulty in deriving approximation bounds (Sections~4-5) is showing the existence of these pseudomatchings for a schedule computed by an algorithm.

\begin{definition}[Bounding graph]
\label{Def:Bounding_Graph}
Let $\mathcal{A}=\{a_1,\ldots,a_k\}$ and $\mathcal{O}=\{o_1,\ldots,o_k\}$ be two equal-cardinality indexed sets of positive real numbers.
We refer to the complete bipartite graph $G=(\mathcal{A}\cup\mathcal{O},\mathcal{A}\times\mathcal{O})$ as the \emph{bounding graph} of $\mathcal{A}$ and $\mathcal{O}$.
\end{definition}

\begin{definition}[$\rho$-pseudomatching]
\label{Def:S_Pseudomatchings}
Given two equal-cardinality indexed sets $\mathcal{A}$ and $\mathcal{O}$ of positive real numbers and their bounding graph $G$, we say that a subset $M\subseteq\mathcal{A}\cup\mathcal{O}$ of edges is a \emph{$\rho$-pseudomatching} if the following properties  hold:
\begin{enumerate}
    \item[4.1] Each node $a_i\in\mathcal{A}$ appears exactly once as an endpoint of a $M$ edge.
    \item[4.2] Each node $o_j\in\mathcal{O}$ appears at most $\rho$ times as an endpoint of a $M$ edge.
    \item[4.3] For each $(a_i,o_j)\in M$, it holds that $a_i\leq o_j$.
\end{enumerate}
\end{definition}

\begin{lemma}
\label{Lem:S_Pseudomachings}
Consider two equal-cardinality indexed sets $\mathcal{A}$ and $\mathcal{O}$ of positive real numbers. 
If the corresponding bounding graph $G$ admits a $\rho$-pseudomatching $M$, then:
\begin{equation*} 
\sum_{a_i\in\mathcal{A}}a_i\leq\rho\left[\sum_{o_j\in\mathcal{O}}o_j\right]. 
\end{equation*}
\end{lemma}
\begin{proof}
Denote by $\mathcal{A}_j=\{a_i:(a_i,o_j)\in M\}$ the subset of $\mathcal{A}$ elements matched with element $o_j\in\mathcal{O}$ in $M$.
Because of Property~4.1, each $a_i$ is matched exactly once, thus $\sum_{a_i\in\mathcal{A}}a_i=\sum_{o_j\in\mathcal{O}}\sum_{a_i\in\mathcal{A}_j}a_i$.
Due to Properties~4.2-4.3, we have that $\sum_{a_i\in\mathcal{A}_j}a_i\leq \rho\cdot o_j$, for each $o_j\in\mathcal{O}$.
Therefore,  we conclude that $\sum_{a_i\in\mathcal{A}}a_i\leq\rho[\sum_{o_j\in\mathcal{O}}o_j]$.
\end{proof}

\begin{definition}[Weak pseudomatching]
\label{Def:G_Pseudomatchings}
Given two equal-cardinality indexed sets $\mathcal{A}$ and $\mathcal{O}$ of positive real numbers and their bounding graph $G$, we say that a subset $M\subseteq\mathcal{A}\cup\mathcal{O}$ of edges is a \emph{weak pseudomatching} if the following  hold:
\begin{enumerate}
    \item[6.1] Each node $a_i\in\mathcal{A}$ appears exactly once as an endpoint of a $M$ edge.
    \item[6.2] For each $(a_i,o_j)\in M$, it holds that $i>j$ and $a_i\leq o_j$.
\end{enumerate}
\end{definition}

\begin{lemma}
\label{Lem:G_Pseudomachings}
Consider two equal-cardinality indexed sets $\mathcal{A}$ and $\mathcal{O}$ of positive real numbers. 
If the corresponding bounding graph $G$ admits a weak pseudomatching $M$, then:
\begin{equation*} 
\sum_{a_i\in\mathcal{A}}(1+\beta)^{n-i}a_i\leq\left(1+\frac{1}{\beta}\right)\left[\sum_{o_j\in\mathcal{O}}(1+\beta)^{n-j}o_j\right]. 
\end{equation*}
\end{lemma}

\begin{proof}
Denote by $\mathcal{A}_j=\{a_i:(a_i,o_j)\in M\}$ the subset of $\mathcal{A}$ elements matched with element $o_j\in\mathcal{O}$ in $M$.
By Property 6.1, it must be the case that $\sum_{a_i\in\mathcal{A}}(1+\beta)^{n-i}a_i=\sum_{o_j\in\mathcal{O}}\sum_{a_i\in\mathcal{A}_j}(1+\beta)^{n-i}a_i$.
Due to Property~6.2, we have that $\sum_{a_i\in\mathcal{A}_j}(1+\beta)^{n-i}a_i\leq\sum_{i=j+1}^n(1+\beta)^{n-i}\max_{a_i\in\mathcal{A}_j}\{a_i\}\leq \sum_{i=j+1}^n(1+\beta)^{n-i}o_j= (1+\frac{1}{\beta})(1+\beta)^{n-j}o_j$, for each $o_j\in\mathcal{O}$, where the last equality follows from a standard geometric series sum calculation.
We conclude that
$\sum_{a_i\in\mathcal{A}}(1+\beta)^{n-i}a_i \leq \left(1+\frac{1}{\beta}\right)\left(\sum_{o_j\in\mathcal{O}}(1+\beta)^{n-j}o_j\right)$.
\end{proof}

\section{Two Basic Algorithms}
\label{Section:Basic_Algorithms}

This section investigates the two greedy \emph{non-interfering} and $\emph{non-idling}$ algorithms that have been proposed for special cases and variants of our problem.
% These algorithms are naturally derived from the makespan expression in Lemma~\ref{Lemma:Delay_Propagation} which decomposes the makespan of a feasible schedule into the sum of a \emph{gap cost} and a \emph{fixed processing time cost}. 
We show that the non-interfering algorithm achieves a constant approximation ratio for instances with $\beta>n+1$, but is $\Omega((1+\beta)^n)$-approximate for general instances.
Next, we argue that the non-idling algorithm attains a constant factor approximation ratio for instances with $\beta\leq1/n$, but is $\Omega((1+\beta)^n)$-approximate for arbitrary instances.
Finally, we prove that returning the best of the two schedules computes a 2-approximate solution for instances with two distinct release times.

\subsection{Non-Interfering Algorithm}

Given a feasible schedule $\mathcal{S}$ for an instance $\mathcal{J}$ of the problem and two jobs $i,j\in\mathcal{J}$, we say that \emph{job $i$ interferes with job $j$} time $t$ in $\mathcal{S}$ if $s_i=t$, $\alpha_i>\alpha_j$ and $t<r_j< (1+\beta)t+\alpha_i$,
i.e.\ the situation where a longer job $i$ begins at a time $t$ before the release time $r_j$ a shorter job $j$ in $\mathcal{S}$ and $i$ completes after $r_j$, which can be avoided with an idle period during $[t,r_j)$. 
Clearly, jobs $i$ and $j$ have start times $s_i<s_j$ in $\mathcal{S}$.
In such a case, we say that $i$ is an \emph{interfering job} in $\mathcal{S}$.
Algorithm~\ref{Algorithm:Non_Interfering} constructs a schedule without interfering jobs.

\begin{algo}[Non-Interfering]
\label{Algorithm:Non_Interfering}
At each time $t$ that the machine becomes available, schedule a pending job $i=\arg\min_{k\in\mathcal{P}(t)}\{\alpha_k\}$ with minimal fixed processing time, unless this job is interfering, i.e.\ there exists a job $j$ such that  $\alpha_j<\alpha_i$ and $t<r_j<(1+\beta)t+\alpha_i$.
In this case, introduce an idle period during $[t,r_j)$ and proceed with time $t=r_j$. 
\end{algo}

Next, we proceed with Lemma~\ref{Lemma:Reduced_Instance} and  Observation~\ref{Observation:Order_To_Schedule}, which simplify the proof of Lemma~\ref{Lemma:Fixed_Processing_Time_Load_Bound} (as we do not need to account for gaps).
Starting from an instance $\mathcal{J}$, Lemma~\ref{Lemma:Reduced_Instance}  defines another instance $\widetilde{\mathcal{J}}$ such that the non-interfering schedules execute the jobs in the same order in the two instances and the non-interfering schedule for $\widetilde{\mathcal{J}}$ does not contain gaps.
Observation~\ref{Observation:Order_To_Schedule} shows that the job order uniquely characterizes an optimal schedule.

\begin{lemma}
\label{Lemma:Reduced_Instance}
Consider an arbitrary instance $\mathcal{J}=\{1,\ldots,n\}$ for which the non-interfering algorithm produces a schedule $\mathcal{S}$ with gaps.
Number the jobs in increasing order $s_1<\ldots<s_n$ of their start times in $\mathcal{S}$.
Starting from $\mathcal{J}$, construct a different instance $\widetilde{\mathcal{J}}$ with the same number of jobs, i.e.\ $|\mathcal{J}|=|\widetilde{\mathcal{J}}|$.
Each job $k\in\widetilde{\mathcal{J}}$ has fixed processing time $\tilde{\alpha}_k=\alpha_k$ and release time $\tilde{r}_k=\min\{r_k,\sum_{i=1}^{k-1}(1+\beta)^{(k-1)-i}\alpha_i\}$, where $\alpha_k$ and $r_k$ are the original parameters of $\mathcal{J}$.
The non-interfering algorithm executes the jobs in the same order in $\mathcal{J}$ and $\widetilde{\mathcal{J}}$, and produces a schedule without gaps, i.e.\ $q_i=0$, for each $i\in\{1,\ldots,n\}$, for $\widetilde{\mathcal{J}}$.
\end{lemma}
\begin{proof}
Starting from $\mathcal{S}$, the new problem instance $\widetilde{\mathcal{J}}$ is constructed so that $|\widetilde{\mathcal{J}}|=|\mathcal{J}|$, by rounding release times down.
In particular, we set a new release time $\tilde{r}_k=\min\{r_k,\sum_{i=1}^{k-1}(1+\beta)^{(k-1)-i}\alpha_i\}$ and fixed processing time $\tilde{\alpha}_k=\alpha_k$, for each $k\in\widetilde{\mathcal{J}}$.
% In addition, we maintain identical fixed processing times identical, i.e.\  for $k\in\widetilde{\mathcal{J}}$.
Consider the schedule $\widetilde{\mathcal{S}}$ for $\widetilde{\mathcal{J}}$ obtained by executing the jobs in the same order with $\mathcal{S}$, but without gaps and denote the makespan of $\widetilde{\mathcal{S}}$ by $\widetilde{T}$. 
% That is, job $k\in\widetilde{\mathcal{J}}$ completes at $\widetilde{C}_k=\sum_{i=1}^k(1+\beta)^{k-i}\alpha_i$ in $\widetilde{\mathcal{S}}$.
By construction, job $k\in\widetilde{\mathcal{J}}$ has start time $\tilde{s}_k=\sum_{i=1}^{k-1}(1+\beta)^{(k-1)-i}\tilde{\alpha}_i\geq \tilde{r}_k$, i.e.\ the new release times are not violated, in $\widetilde{\mathcal{S}}$.
Next, consider any pair $k,\ell\in\widetilde{\mathcal{J}}$ of jobs such that $k<\ell$ and $\tilde{\alpha}_k>\tilde{\alpha}_{\ell}$.
Since $\mathcal{S}$ is non-interfering, we have that $C_k<r_{\ell}$, which implies that $\sum_{i=1}^{k-1}(1+\beta)^{(k-1)-i}\alpha_i\leq r_{\ell}$.
Given that $k<\ell$, we conclude that $\tilde{C}_k\leq\tilde{r}_{\ell}$, i.e.\ $\widetilde{\mathcal{S}}$ is non-interfering for $\widetilde{\mathcal{J}}$.
\end{proof}

\begin{observation}
\label{Observation:Order_To_Schedule}
Consider an arbitrary order $\gamma(\cdot)$ of the jobs, where $\gamma(k)\in\mathcal{J}$ is the job in the $k$-th position of the order, for $k\in\{1,\ldots,n\}$.
Suppose that there exists an optimal schedule $\mathcal{S}^*$ executing the jobs according to $\gamma(\cdot)$, i.e.\ $s_{\gamma(1)}^*<\ldots<s_{\gamma(n)}^*$, where $s_i^*$ is the start time of job $i\in\mathcal{J}$ in $\mathcal{S}^*$. 
W.l.o.g.\ it holds that $s_{\gamma(i)}=\max\{r_{\gamma(i)},C_{\gamma(i-1)}\}$, for $i\in\{1,\ldots,n\}$, where $C_{\gamma(0)}=0$.
\end{observation}

Number the jobs in increasing order $s_1<\ldots<s_n$ of their start times in the schedule $\mathcal{S}$ produced by the non-interfering algorithm.
Job $k\in\mathcal{J}$ is executed in the $k$-th position of $\mathcal{S}$.
Next, consider an optimal schedule $\mathcal{S}^*$ and let $\gamma(k)\in\mathcal{J}$ be the job executed in the $k$-th position of $\mathcal{S}^*$, for $k\in\{1,\ldots,n\}$.
We say that $k\in\mathcal{J}$ is a \emph{critical job} if $C_k\leq C_{\gamma(k)}^*$, where $C_i^*$ is the completion time of job $i\in\mathcal{J}$ in $\mathcal{S}^*$. 
% That is, a critical job $k$ has a completion time $C_k$ in $\mathcal{S}$ which is not greater than the completion time $C_{\gamma(k)}^*$ of job $\gamma(k)$ executed in the $k$-th position of $\mathcal{S}^*$.
Lemma~\ref{Lemma:Fixed_Processing_Time_Load_Bound} upper bounds the fixed processing times of jobs executed after the last critical job in $\mathcal{S}$ based on pseudomatchings.

\begin{lemma}
\label{Lemma:Fixed_Processing_Time_Load_Bound}
Consider a non-interfering schedule $\mathcal{S}$ and let $\ell=\max\{k:C_k\leq C_{\gamma(k)}^*, k\in\mathcal{J}\}$ be the last critical job.
For each $k\in\{\ell+1,\ldots,n\}$, it holds that $\sum_{i=\ell+1}^k\alpha_i\leq2\left[\sum_{j=1}^k\alpha_{\gamma(j)}\right]$.
\end{lemma}
\begin{proof}
To prove the lemma, we may assume w.l.o.g.\ that $\mathcal{S}$ does not contain gaps.
Otherwise, if $\mathcal{S}$ contains gaps, starting from the original instance $\mathcal{J}$, we may consider the modified instance $\widetilde{\mathcal{J}}$ obtained according to Lemma~\ref{Lemma:Reduced_Instance}.
Using Observations~\ref{Observation:Order_To_Schedule} and the orders of the jobs in the non-interfering schedule $\mathcal{S}$ and in an optimal schedule $\mathcal{S}^*$ for $\mathcal{J}$, we can obtain two feasible schedules $\widetilde{\mathcal{S}}$ and $\widetilde{\mathcal{S}}^*$, respectively, for $\widetilde{\mathcal{J}}$. 
By Lemma~\ref{Lemma:Reduced_Instance}, $\widetilde{\mathcal{S}}$ is the schedule produced by the non-interfering algorithm for $\widetilde{\mathcal{J}}$ and does not contain any gaps.
Therefore, proving the lemma with $\widetilde{\mathcal{S}}$ and $\widetilde{\mathcal{S}}^*$ implies that the lemma holds for the original instance $\mathcal{J}$.
We note that the optimal schedule may change for $\widetilde{\mathcal{J}}$, but this does not affect our argument since we compare a non-interfering schedule without gaps with an arbitrary feasible schedule.
In the remainder of the proof, assume that $q_i=0$, for each $i\in\mathcal{J}$, in $\mathcal{S}$.

% Consider an arbitrary order $\gamma$ of the jobs, where $\gamma(k)\in\mathcal{J}$ is the job in $k$-th position of the order.
% Any order $\gamma$ can be easily converted into a feasible for $\mathcal{J}$ such that $s_{\gamma(1)}<\ldots<s_{\gamma(n)}$.
% It suffices to set $s_{\gamma(i)}=\max\{r_{\gamma(i)},C_{\gamma(i-1)}\}$, for $i=1,\ldots,n$, where $C_{\gamma(0)}=0$.
% Clearly, a schedule produced by the non-interfering algorithm and an optimal schedule satisfy this expression.
% Based on this observation, we may assume without loss of generality that $\mathcal{S}$ does not contain gaps.
% Otherwise, given the order $\gamma(\cdot)$ of the jobs in $\mathcal{S}$ and the order $\gamma^*(\cdot)$ of the jobs in $\mathcal{S}^*$, we can consider the modified instance $\widetilde{\mathcal{J}}$ obtained according to Lemma~\ref{Lemma:Reduced_Instance} and derive the inequality using this instance and the schedules $\widetilde{\mathcal{S}}$ and $\widetilde{\mathcal{S}}^*$ fro $\widetilde{\mathcal{J}}$ obtained as described above with the orders $\gamma(\cdot)$ and $\gamma^*(\cdot)$, respectively. 
% by arguing using the feasible schedule for $\widetilde{\mathcal{J}}$ which executes the jobs in the same order with an optimal schedule $\mathcal{S}^*$ for $\mathcal{J}$. 

For each $k\in\{1,\ldots,n\}$, define the sets $\mathcal{A}_k=\{1,\ldots,k\}$ and $\mathcal{O}_k=\{\gamma(1),\ldots,\gamma(k)\}$ of jobs executed in the first $k$ positions of $\mathcal{S}$ and $\mathcal{S}^*$, respectively.
Further, for each $k>\ell$, denote by $\mathcal{A}_k^-=\{\ell+1,\ldots,k\}$ the subset of the $\mathcal{A}_k$ jobs executed after the last critical job $\ell$ in $\mathcal{S}$.
For simplicity of the presentation, we denote a job in $\mathcal{A}_k$ by its actual index $i$ and a job in $\mathcal{O}_k$ by $\gamma(j)$ (i.e.\ using the $\gamma(\cdot)$ notation), for $i,j\in\{1,\ldots,k\}$.
Deriving the lemma is equivalent to showing that $\sum_{i\in\mathcal{A}_k^-}\alpha_i\leq2\left[\sum_{\gamma(j)\in\mathcal{O}_k}\alpha_{\gamma(j)}\right]$.

For each $k\in\{1,\ldots,n\}$, consider the bounding (complete bipartite) graph $G_k=(\mathcal{A}_k\cup \mathcal{O}_k,\mathcal{A}_k\times\mathcal{O}_k)$ with $2k$ nodes: a node for each of the $k$ jobs in $\mathcal{A}_k$ and a node for each of the $k$ jobs in $\mathcal{O}_k$.
Note that, if there exist $i\in\mathcal{A}_k$ and $\gamma(j)\in\mathcal{O}_k$ such that $i=\gamma(j)$, then we introduce two nodes for job $i$, i.e.\ a node in each side of the bipartition. 
% Note that we introduce a node in each side of the bipartition (i.e.\ two nodes in total) for each job $i\in\mathcal{A}_k\cap\mathcal{O}_k$ that belongs to both $\mathcal{A}_k$ and $\mathcal{O}_k$.
The graph contains all possible $k^2$ edges with one endpoint in $\mathcal{A}_k$ and the other in $\mathcal{O}_k$.
Using standard terminology, a matching in $G_k$ is a subset $M_k\subseteq\mathcal{A}_k\times\mathcal{O}_k$ of edges without a common endpoint.
If $(i,\gamma(j))\in M_k$, then we say that the nodes $i\in\mathcal{A}_k$ and $\gamma(j)\in\mathcal{O}_k$ are \emph{matched} by $M_k$.
% Further, we say that $i$ has been \emph{covered} by $M_k$, since there is an edge containing it.
% That is, each node $i\in\mathcal{A}_k$ appears at most once as an endpoint of an edge in $M_k$.
By relaxing the notion of a matching, we refer to a set $M_k$ of edges in $G_k$ as a $\rho$\emph{-pseudomatching} if every node $i\in\mathcal{A}_k$ appears at most once as an endpoint of an edge in $M_k$ and every node $\gamma(j)\in\mathcal{O}_k$ appears at most $\rho$ times as an edge point of an edge in $M_k$, where $\rho\in\mathbb{Z}^+$ is a positive integer.
% Further, $M_k$ is a \emph{critically maximal matching} if it contains $|M_k|=n-k$ edges, so that there exists an edge $(i,\gamma(j))$ containing $i$ as an endpoint, for each $i\in\mathcal{A}_k$. 
Let $M_k(\mathcal{A}_k)=\{i:(i,\gamma(j))\in M_k, i\in\mathcal{A}_k, \gamma(j)\in\mathcal{O}_k\}$ be the subset of the $\mathcal{A}_k$ nodes appearing as an endpoint of an edge in $M_k$.
Next, for each $k\in\{\ell+1,\ldots,n\}$, we show the existence of a 2-pseudomatching $M_k$ in $G_k$ with the following properties:
% We say that $M_k$ is a \emph{dominating $\rho$-pseudo-matching} of $G_k$, if (a) it is a $\rho$-pseudo-matching, (b) $\mathcal{A}_k^-\subseteq M_k(\mathcal{A}_k)$, and (c) $\alpha_i\leq\alpha_{\gamma(j)}$ for each $(i,\gamma(j))\in M_k$.
% It suffices to show the existence of a pseudo-matching with the following properties:
\begin{enumerate}
    \item $M_k(\mathcal{A}_k)=\mathcal{A}_k^-$, i.e.\ each job $i\in\mathcal{A}_k^-$ appears exactly once as the endpoint of an edge in $M_k$ and no other $\mathcal{A}_k$ node is matched.
    \item For every job $i\in\mathcal{A}_k^-$ such that there exists a job $\gamma(j)\in\mathcal{O}_k$ with $i=\gamma(j)$, we have that $(i,\gamma(j))\in M_k$. That is, every job $i$ which is executed in the $\{\ell+1,\ldots,k\}$ positions of $\mathcal{S}$ and the first $k$ positions of $\mathcal{S}^*$ must be matched with itself in $M_k$.
    \item Every job $i\in\mathcal{A}_k^-$ which is not executed in the first $k$ positions of $\mathcal{S}^*$, i.e.\ $\gamma(j)=i$ for some $\gamma(j)\notin \mathcal{O}_k$, must be matched with a job  $\gamma(j)\in\mathcal{O}_k\setminus\mathcal{A}_k$ in $M_k$.
    \item Each job $\gamma(j)\in\mathcal{O}_k\setminus\mathcal{A}_k$ is matched with at most one job in $\mathcal{A}_k\setminus\mathcal{O}_k$.
    \item If $(i,\gamma(j))\in M_k$, for some pair of jobs $i\in\mathcal{A}_k$ and $\gamma(j)\in\mathcal{O}_k$, then $\alpha_i\leq\alpha_{\gamma(j)}$.
\end{enumerate}

We refer to a pseudomatching satisfying the above properties as a \emph{2-pseudomatching}.
If such a pseudomatching exists, then it clearly holds that $\sum_{i\in\mathcal{A}_k^-}\alpha_i\leq2\left[\sum_{\gamma(j)\in\mathcal{O}_k}\alpha_{\gamma(j)}\right]$: each $\mathcal{A}_k^-$ job is matched exactly once with an $\mathcal{O}_k$ job and each $\mathcal{O}_k$ job is matched at most two $\mathcal{A}_k^-$ jobs.
% In the extreme a case, a job $\gamma(j)\in\mathcal{O}_k$ is matched with itself, i.e.\ $(i,\gamma(j))\in M_k$ and $i=\gamma(j)$, and with a job $i'\in\mathcal{A}_k^-\setminus\mathcal{O}_k$, where $\gamma(j)\in\mathcal{O}_k\setminus\mathcal{A}_k$.
We will show its existence by induction on $k\in\{\ell+1,\ldots,n\}$.

% and derive $M_k$ so that it satisfies the following \emph{identity property}: for each job $i\in\mathcal{A}_k\cap\mathcal{O}_k$, i.e.\ there exists $j\in\{1,\ldots,k\}$ such that $\gamma(j)=1$, it holds that $(i,\gamma(j))\in M_k$.

% To prove the lemma, if suffices to show that there exists a dominating 2-pseudo-matching $M_k$ in $G_k$, for each $k\in\{\ell+1,\ldots,n\}$.
% We show this statement by induction on $k\in\{\ell+1,\ldots,n\}$. 

For the induction basis, consider the case $k=\ell+1$.
If $\ell+1\in\mathcal{O}_k$, then $\gamma(j)=\ell+1$, for some $j\in\{1,\ldots,\ell+1\}$.
Clearly, $\mathcal{M}_{\ell+1}=\{(\ell+1,\gamma(j))\}$ is a 2-pseudomatching, given that $\alpha_{\ell+1}=\alpha_{\gamma(j)}$.
If $\ell+1\notin\mathcal{O}_k$, then, by using a simple pigeonhole principle argument\footnote{A similar, but more elaborate, pigeonhole argument is rigorously presented in the proof of Theorem~\ref{Theorem:ECTF}.}, there exists a job $\gamma(j)>\ell+1$ such that $j\in\{1,\ldots,\ell+1\}$.
Since $\ell+1$ is not critical, we have that $\alpha_{\ell+1}\leq\alpha_{\gamma(j)}$.
Otherwise, if $\alpha_{\ell+1}>\alpha_{\gamma(j)}$, by the way the non-interfering algorithm works and the fact that $\ell+1<\gamma(j)$, we would have $C_{\ell+1}\leq r_{\gamma(j)}<C_{\gamma(\ell+1)}^*$, which would contradict that $\ell+1$ is not critical.
We conclude that $M_{\ell+1}=\{(\ell+1,\gamma(j))\}$ is a 2-pseudomatching.

% We refer the reader to Figure \ref{fig:match} for the schematic representation of the proof.
For the induction step, assume that $G_k$ admits a 2-pseudomatching $M_k$.
We will convert $M_k$ into a 2-pseudomatching $M_{k+1}$ for $G_{k+1}$.
This update involves the following steps:
\begin{itemize}
    \item Initially, we adapt $M_k$ based on the job $\gamma(k+1)$ executed in the $(k+1)$-th position of $\mathcal{S}^*$ to obtain an intermediate 2-pseudomatching $\widetilde{M}_{k+1}$. 
    Suppose that $\gamma(k+1)=i$.  
    If $i\notin\mathcal{A}_k^-$, then we set $\widetilde{M}_{k+1}=M_k$.
    Otherwise, if $i\in\mathcal{A}_k^-$, i.e.\  $\mathcal{S}$ completes job $i$ in the positions $\{\ell+1,\ldots,k\}$, then we need to update $M_k$ so as to satisfy Properties 1-2 in the resulting 2-pseudomatching $M_{k+1}$.
    Since $\gamma(k+1)\notin\mathcal{O}_k$, by the induction hypothesis, job $i$ is matched with exactly one job $\gamma(j)\in\mathcal{O}_k\setminus\mathcal{A}_k$ in $M_k$.
    We set $\widetilde{M}_{k+1}=(M_k\cup\{(i,\gamma(k+1))\})\setminus\{(i,\gamma(j))\}$, that is we remove $(i,\gamma(j))$ from $M_k$ and add  $(i,\gamma(k+1))$ to obtain $\widetilde{M}_{k+1}$.
    In this way, $i\in\mathcal{A}_k^-$ is now matched with job $\gamma(k+1)$, i.e.\ itself, in the $\mathcal{O}_k$ side of $G_k$.
    
    \item Next, we adapt $\widetilde{M}_{k+1}$ so as to have job $k+1$ matched with some $\mathcal{O}_{k+1}$ job in $M_{k+1}$. We distinguish two cases based on whether $(k+1)\in\mathcal{O}_{k+1}$, or $(k+1)\notin\mathcal{O}_{k+1}$.
    In the former case, there exists $j\in\{1,\ldots,k+1\}$ s.t.\ $\gamma(j)=k+1$. Based on Property 2, we set $M_{k+1}=\widetilde{M}_{k+1}\cup\{(k+1,\gamma(j))\}$, i.e.\ job $(k+1)\in\mathcal{A}_{k+1}^-$ is matched with itself in the $\mathcal{O}_{k+1}$ side.
    In the latter case, it holds that $k+1\in\mathcal{A}_{k+1}\setminus\mathcal{O}_{k+1}$.
    Let $x=|\mathcal{A}_{k+1}\setminus\mathcal{O}_{k+1}|$ be the number of jobs executed in the first $k+1$ positions of $\mathcal{S}$ and after the first $k+1$ positions of $\mathcal{S}^*$.
    A simple set theoretic argument implies that $|\mathcal{O}_{k+1}\setminus\mathcal{A}_{k+1}|=x$. 
    By the induction hypothesis and Properties 3-4, we conclude that each $\mathcal{O}_k\setminus\mathcal{A}_k$ is matched with at most one $\mathcal{A}_k\setminus\mathcal{O}_k$ job.
    Therefore, given that $i\in\mathcal{A}_{k+1}\setminus\mathcal{O}_{k+1}$, there exists a job $\gamma(j)\in\mathcal{O}_{k+1}\setminus\mathcal{A}_{k+1}$ which is not matched with any $\mathcal{A}_k$ job in $\widetilde{M}_{k+1}$.
    We set $M_{k+1}=\widetilde{M}_{k+1}\cup\{(i,\gamma(j))\}$ and guarantee that Properties 1-4 are satisfied.
    Next, we claim that $\alpha_{k+1}\leq\alpha_{\gamma(j)}$. Otherwise, if $\alpha_{k+1}>\alpha_{\gamma(j)}$, by the way the non-interfering algorithm works and the fact that $k+1<\gamma(j)$ ($\gamma(j)\in\mathcal{O}_{k+1}\setminus\mathcal{A}_{k+1}$), we would have $C_{k+1}\leq r_{\gamma(j)}<C_{\gamma(k+1)}^*$, which would contradict that $k+1$ is not critical.
\end{itemize}
\end{proof}

Lemma~\ref{Lemma:Release_Times_Lower_Bound} lower bounds the optimal makespan using release times. 

\begin{lemma}
\label{Lemma:Release_Times_Lower_Bound}
Assume that the jobs are numbered so that $r_1\leq\ldots\leq r_n$. Any optimal schedule $\mathcal{S}^*$ has makespan $T^*\geq \sum_{i=1}^n\beta^{n-i}r_i$.
\end{lemma}
\begin{proof}
Denote by $C_i$ and $T$ the completion time of job $i\in\mathcal{J}$ and makespan, respectively, in schedule $S$.
We prove by induction that $C_k\geq\sum_{i=1}^k\beta^{k-i}r_i$, for each $k\in\{1,\ldots,n\}$.
For $k=1$, it clearly holds that $C_1\geq \beta s_1^*\geq r_1$.
Suppose that lemma is true for some $k\in\{1,\ldots,n-1\}$.
By the fact that $s_{k+1}\geq\max\{r_{k+1},C_k\}$ and the induction hypothesis:
\begin{equation*}
C_{k+1}= (1+\beta)s_{k+1}+\alpha_{k+1}\geq r_{k+1}+\beta C_k\geq r_{k+1}+\beta\left[\sum_{i=1}^k\beta^{k-i}r_i\right] 
= \sum_{i=1}^{k+1}\beta^{(k+1)-i}r_i.
\end{equation*}
\end{proof}

Theorem~\ref{Theorem:Non_Interfering} presents bounds on the approximation ratio of the non-interfering algorithm.

\begin{theorem}
\label{Theorem:Non_Interfering}
Algorithm~\ref{Algorithm:Non_Interfering} is $(3+e)$-approximate for instances with $\beta\geq n+1$ and $\Omega((1+\beta)^n)$-approximate for general instances.
\end{theorem}
\begin{proof}
% We prove the upper and lower bounds individually.
% We begin with the upper bound for instances with $\beta>n+1$ and proceed with the lower bound for general instances.
% \textbf{Upper Bound.}
% We prove the upper bound here and show the lower bound in the Appendix.
% Consider an instance $\mathcal{J}$ for which the non-interfering algorithm produces a schedule $\mathcal{S}$ without gaps, i.e.\ $q_i=0$, for each $i\in\mathcal{J}$. 
% Then, $\mathcal{S}$ has makespan $T\leq 2T^*$, where $T^*$ is the optimal makespan for $\mathcal{J}$.
Recall that the jobs are numbered in increasing order $s_1<\ldots\leq s_n$ of their start times in the schedule $\mathcal{S}$ produced by the non-interfering algorithm
% , i.e.\ job $k\in\mathcal{J}$ is executed in the $k$-th position of $\mathcal{S}$, 
and $\gamma(k)\in\mathcal{J}$ is the job executed in the $k$-th position of an optimal schedule $\mathcal{S}^*$, for $k\in\{1,\ldots,n\}$.
Let $\ell=\max\{k:C_k\leq C_{\gamma(k)}^*
\}$ be the last critical position.
By Lemma~\ref{Lemma:Delay_Propagation}, we have that $T = (1+\beta)^{n-\ell}C_{\ell}+\sum_{i=\ell+1}^n(1+\beta)^{n-i+1}q_i+\sum_{i=\ell+1}^n (1+\beta)^{n-i}\alpha_i$.
Using Lemma~\ref{Lemma:New_Makespan_Expression}, i.e.\ expanding the last sum of this expression with geometric series, we get that:
\begin{align}
\label{Equation:Non_Interfering_Makespan_Expression}
T & = (1+\beta)^{n-\ell}C_{\ell} + \sum_{i=\ell+1}^n(1+\beta)^{n-i+1}q_i
+\sum_{i=\ell+1}^n\alpha_i+\sum_{k=\ell+2}^n\beta(1+\beta)^{n-k}\left(\sum_{i=\ell+1}^{k-1} \alpha_i\right)
\end{align}
Consider job $i\in\mathcal{J}$.
If $q_i>0$, then job $i$ begins at its release time $r_i$.
That is, the gap of length $q_i$ immediately preceding job $i$ occurs exactly during the time interval $[r_i-q_i,r_i)$.
Hence, $q_i\leq r_i$.
Based on this observation, the obvious fact that $\sum_{i=\ell+1}^n\alpha_i\leq\sum_{i=1}^n\alpha_{\gamma(i)}$ and Lemma~\ref{Lemma:Fixed_Processing_Time_Load_Bound}, Equation~(\ref{Equation:Non_Interfering_Makespan_Expression}) implies that:
\begin{align}
\label{Equation:Non_Interfering_Upper_Bound}
T \leq (1+\beta)^{n-\ell}C_{\gamma(\ell)}^* + 
\sum_{i=\ell+1}^n\left(1+\frac{1}{\beta}\right)^{n-i+1}\beta^{n-i+1}r_i \nonumber \\
 + 2\left[\sum_{i=1}^n\alpha_{\gamma(i)} + \sum_{k=\ell+2}^n\beta(1+\beta)^{n-k}\left(\sum_{i=1}^{k-1}\alpha_{\gamma(i)}\right) \right]
\end{align}

By the definition of $\gamma(\cdot)$, Lemma~\ref{Lemma:New_Makespan_Expression} and Lemma~\ref{Lemma:Release_Times_Lower_Bound}, we get that
\begin{align}
\label{Equation:Non_Interfering_Lower_Bound}
T^*\geq\max\left\{(1+\beta)^{n-\ell}C_{\gamma(\ell)}^*, \sum_{j=1}^n\alpha_{\gamma(j)}+\sum_{k=2}^{n}\beta(1+\beta)^{n-k}\left(\sum_{i=1}^{k-1}\alpha_{\gamma(i)}\right), \sum_{i=1}^n\beta^{n-i+1}r_i\right\}
\end{align}

For $\beta\geq n+1$, we have that $(1+\frac{1}{\beta})^{n+1}\leq e$.
Therefore, Equations (\ref{Equation:Non_Interfering_Upper_Bound})-(\ref{Equation:Non_Interfering_Lower_Bound}) imply that $T\leq (3+e)T^*$.

For the lower bound, consider an instance with $n$ jobs, where $r_{\min}=B$, for some large constant $B=\omega(n)$.
Job $j\in\{1,\ldots,n\}$ has $\alpha_j=B+n-j$ and $r_j=\sum_{i=1}^j(1+\beta)^{i-1}B$. 
%where job $j\in\{1,\ldots,n\}$ has fixed processing time $\alpha_j=B+n-j$ and release time $r_j=\sum_{i=1}^j(1+\beta)^{i-1}B$. %and $B=\omega(n)$.
We show by induction on $j$ that no job begins before $r_j$ in the algorithm's schedule $\mathcal{S}$.
Given that $r_1<\ldots< r_n$, our claim trivially holds for $j=1$ because no job can be executed before $r_1=\min_{i\in\mathcal{J}}\{r_i\}$.
%Moreover, if job $1$ begins at $r_1$, then it completes at time $t=B+(n-1)>B=r_2$.
%Hence, the non-interfering algorithm will leave an idle period during $[r_1,r_2)$.
For the induction hypothesis, assume that our claim is true for some $j\geq 1$, i.e.\ no job begins before $r_j$ in $\mathcal{S}$.
Since $\alpha_1\geq\ldots\geq\alpha_j$, any job beginning at time $r_j$ would have completion time at least:
\begin{align*}
(1+\beta)r_j+\alpha_j\geq \sum_{i=1}^{j}(1+\beta)^{i-1}B + n-j > r_{j+1}.
\end{align*}
So, the algorithm will not schedule any job during $[r_j,r_{j+1})$, because otherwise this job would be interfering. 
Our claim implies that the algorithm schedules all jobs according to shortest fixed processing time first starting from $r_n$ and has makespan: 
\begin{equation*}
T=(1+\beta)^nr_n+\sum_{j=1}^n(1+\beta)^{n-j}B+\sum_{j=1}^n(1+\beta)^{n-j}(j-1) = \Omega((1+\beta)^{2n}B).
\end{equation*}
In an optimal schedule $\mathcal{S}^*$, all jobs begin at time $t=0$ and are consecutively executed according without any idle period between them according to earliest release time first. 
The makespan of $\mathcal{S}^*$ is:
\begin{align*}
T^*=\sum_{j=1}^n(1+\beta)^{n-j}B+\sum_{j=1}^n(1+\beta)^{n-j}(n-j) = O((1+\beta)^n).
\end{align*}
% Since $r_n=\sum_{j=1}^{n-1}(1+\beta)^i\leq\frac{(1+\beta)^n}{\beta}$, we conclude that $\frac{T}{T^*}=\Omega((1+\beta)^n)$.
Hence, $T/T^*=\Omega((1+\beta)^n)$.
\end{proof}

\subsection{Non-Idling Algorithm}

Algorithm~\ref{Algorithm:No_Gap} constructs a feasible schedule by executing the shortest pending job whenever the machine becomes available.

\begin{algo}[Non-Idling]
\label{Algorithm:No_Gap}
Greedily schedule jobs over time, by initiating a pending job $\arg\min_{i\in\mathcal{P}(t)}\{\alpha_i\}$ with minimal fixed procesing time at each time $t$ that the machine becomes available.
\end{algo}

% Lemma \ref{Lemma:No_Gap_Dependent_Term_Bound} shows that the no-gap algorithm achieves a low gap-dependent cost.

\begin{theorem}
\label{Theorem:Non_Gap}
Algorithm~\ref{Algorithm:No_Gap} is $(1+e)$-approximate for instances with $\beta\leq 1/n$ and $\Omega((1+\beta)^n)$-approximate for general instances.
\end{theorem}

\begin{proof}
On the positive side, consider a schedule $\mathcal{S}$ produced by the non-idling algorithm and number the jobs in increasing order $s_1<\ldots<s_n$ of their start times in $\mathcal{S}$.
Let $Q=\sum_{i=1}^n(1+\beta)^{n-i+1}q_i$ and $A=\sum_{i=1}^n(1+\beta)^{n-i}\alpha_i$ be the gap-dependent and fixed processing time costs of $\mathcal{S}$, respectively.
Next, we will show that $Q\leq T^*$ and $A\leq e T^*$, where $T^*$ is the optimal makespan.
By Lemma~\ref{Lemma:Delay_Propagation}, we get that $T=Q+A\leq(1+e)T^*$.

To bound the gap-dependent cost of the algorithm, we show the existence of an optimal schedule $\mathcal{S}^*$ satisfying the property that, for each idle time interval $[t,u)$ in $\mathcal{S}$, the interval $[t,u)$ is also idle in $\mathcal{S}^*$. 
For simplicity, we prove the lemma for the case where $\mathcal{S}$ contains a single maximal idle time interval, i.e.\ gap $q_j>0$, but the argument is naturally extended to an arbitrary number of gaps.
We may partition $\mathcal{J}$ into the sets $\mathcal{J}_A=\{i\in\mathcal{J}:s_i\geq t\}$ and $\mathcal{J}_B=\{i\in\mathcal{J}:s_i<t\}$ of jobs beginning after and before, respectively, time $t$ in $\mathcal{S}$.
Using this definition and the fact that $[t,u)$ is idle in the algorithm's schedule, we conclude that $r_i\geq u$ for each $i\in\mathcal{J}_A$ and $r_i<t$ for every $i\in\mathcal{J}_B$. 
Let $\mathcal{S}_A^*$ and $\mathcal{S}_B^*$ be optimal schedules for $\mathcal{J}_A$ and $\mathcal{J}_B$ of makespans $T_A^*$ and $T_B^*$, respectively.
Clearly, the schedule $\mathcal{S}^*$ obtained by merging $\mathcal{S}_A^*$ and $\mathcal{S}_B$ is feasible and optimal for $\mathcal{J}$ given that $T^*=T_A^*$. 
Thus, $Q\leq T^*$.
To bound the algorithm's fixed processing time cost, by using the standard Euler constant inequality $(1+\frac{1}{k})^k\leq e$, for each constant $k\geq 1$, we get that: 
\begin{align*}
A=\sum_{i=1}^n(1+\beta)^{n-i}\alpha_i\leq
(1+\beta)^n\left[\sum_{i=1}^n\alpha_i\right]\leq
e\left[\sum_{i=1}^n\alpha_i\right]\leq e T^*.
\end{align*}

On the negative side, consider an instance with $n=k+1$ jobs, namely a job of fixed processing time $\alpha_1=B>1$, release time $r_1=0$, and $k$ jobs with $\alpha_i=0$ and $r_i=1$, for $i\in\{2,\ldots,k+1\}$.
The non-idling schedule executes the jobs in increasing order of their indices, i.e.\ the first job completes at $C_1=B$ and all remaining jobs are consecutively executed starting at $C_1$. 
By Lemma~\ref{Lemma:Delay_Propagation}, $\mathcal{S}$ has makespan $T=(1+\beta)^kB$.
In an optimal schedule $\mathcal{S}^*$, all short jobs are consecutively executed during $[1,(1+\beta)^k]$, the long job begins right after and completes at $T^*=(1+\beta)^{k+1}+B$. 
If $B=(1+\beta)^{k+1}$,
\begin{align*}
\frac{T}{T^*}=\frac{(1+\beta)^kB}{(1+\beta)^{k+1}+B}=\Omega((1+\beta)^k).
\end{align*}
\end{proof}

\subsection{Best-of-Two Algorithm}

The \emph{best-of-two} algorithm returns the best among the non-idling and non-interfering schedules. 
Theorem~\ref{Theorem:Two_Release_Times_Constant} shows that this algorithm achieves a 2-approximation ratio when $r_i=\{0,r\}$ for each job $i\in\mathcal{J}$.

\begin{theorem}
\label{Theorem:Two_Release_Times_Constant}
The best-of-two algorithm is 2-approximate for instances with two distinct release times.
\end{theorem}
\begin{proof}
Consider an optimal schedule $\mathcal{S}^*$ of makespan $T^*$, in which job $i\in\mathcal{J}$ begins at $s_i^*$ and completes at $C_i^*$.
Further, denote by $k^*=|\{j\in\mathcal{J}:s_j^*<r\}|$ the number of jobs beginning before $r$ and let $\gamma$ be the order $s_{\gamma(1)}^*<\ldots<s_{\gamma(n)}^*$ of jobs in increasing start times, i.e.\ job $\gamma(i)$ is executed in the $i$-th position of $\mathcal{S}^*$.
For a given subset $\mathcal{J}'=\{\pi(1),\ldots,\pi(k)\}$ of $k$ jobs which are numbered so that $\alpha_{\pi(1)}\leq\ldots\leq\alpha_{\pi(k)}$, denote by $F(\mathcal{J}')=\sum_{i=1}^k(1+\beta)^{k-i}\alpha_{\pi(i)}$ their fixed processing time cost if they are continuously scheduled without gaps and other intermediate jobs in non-decreasing order of their fixed processing times.
We distinguish two cases based on whether $C_{\gamma(k^*)}^*\leq r$ and $C_{\gamma(k^*)}^*>r$.

In the former case, since $n-k^*$ jobs begin after $r$ in $\mathcal{S}^*$, Lemma~\ref{Lemma:Delay_Propagation} implies that $T^*\geq\max\{(1+\beta)^{n-k^*}r,F(\mathcal{J})\}$.
Assume that the algorithm's non-interfering schedule $\mathcal{S}$ has makespan $T$ and suppose that it associates a start time $s_j$ and completion time $C_j$ to each job $j\in\mathcal{J}$.
Also, let $k=|\{j\in\mathcal{J}:s_j<r\}|$.
We claim that $k\geq k^*$.
Assume for contradiction that $k<k^*$.
W.l.o.g.\ we may assume that $\alpha_{\gamma(1)}\leq\ldots\leq\alpha_{\gamma(k^*)}$, i.e.\ $\mathcal{S}^*$ schedules jobs in non-decreasing order of fixed processing times before $r$.
Because $\mathcal{S}$ schedules the pending job with the shortest fixed processing time at each time that the machine becomes available, it must be the case that $\alpha_i\leq\alpha_{\gamma(i)}$, for $1\leq i\leq k$.
If $C_k<C_{\gamma(k^*)}^*$, then there exists a job $j\in\mathcal{J}$ %$j\in\{\pi(1),\ldots,\pi(k^*)\}$ 
such that $s_j^*<r\leq s_j$ which can be feasibly executed during $[C_k,r]$ in $\mathcal{S}$, i.e.\ a contradiction on the definition of $k$. 
If $C_k\geq C_{\gamma(k^*)}^*$, then there exist jobs $i,j\in\mathcal{J}$ such that $\alpha_i>\alpha_j$, $s_i<r\leq s_j$ and 
$s_i^*\geq r>s_{j}^*$,
which contradicts the fact that the algorithm always schedules a pending job with a minimal processing time.
Hence, our claim is true.
By Lemma~\ref{Lemma:Delay_Propagation}, if $\mathcal{J}'=\{j\in\mathcal{J}:s_j\geq r\}$, then 
$T = (1+\beta)^{n-k}r+F(\mathcal{J}')
\leq (1+\beta)^{n-k^*}r+F(\mathcal{J})
\leq 2T^*$.

In the latter case, consider the non-idling schedule $\mathcal{S}$ of the algorithm, denote its makespan by $T$ and the execution interval of each job $j\in\mathcal{J}$ by $[s_j,C_j]$.
Let $t^*=\max_{j\in\mathcal{J}}\{C_j^*:s_j^*<r\}$ and $t=\max_{j\in\mathcal{J}}\{C_j:s_j<r\}$ be the completion time of the interfering job in $\mathcal{S}^*$ and $\mathcal{S}$, respectively.
Similarly before, $T^*\geq\max\{(1+\beta)^{n-k^*}t^*,F(\mathcal{J})\}$.
Given that $\mathcal{S}$ executes the jobs with the minimal fixed processing times until it encounters job $k$ with $C_k>r$, we have that $k\geq k^*$.
If $t\leq t^*$, then 
$T = (1+\beta)^{n-k}t+F(\mathcal{J}') \leq (1+\beta)^{n-k^*}t^*+F(\mathcal{J})
\leq 2T^*$, where $\mathcal{J}'=\{j\in\mathcal{J}:s_j\geq t\}$.
If $t>t^*$, then $k\geq k^*-1$, i.e.\ $T^*\geq(1+\beta)^{n-k-1}r$. 
Therefore, $T\leq (1+\beta)^{n-k-1}r+F(\mathcal{J}'\cup\{k\})\leq 2T^*$.
\end{proof}

\section{Earliest Completion-Time First}
\label{Section:ECTF}

Next, we consider the \emph{Earliest Completion Time First (ECTF)} algorithm and show that it is $O(1+\frac{1}{\beta})$-approximate.
ECTF produces a schedule satisfying Observation~\ref{Observation:Order_To_Schedule}.
That is, if jobs are numbered in increasing order $s_1<\ldots<s_n$ of their start times in $\mathcal{S}$, then job $i\in\mathcal{J}$ has start time $s_i=\max\{r_i,C_{i-1}\}$.
At every time $t$, let $\Gamma_{i}(t)=(1+\beta)\max\{t,r_{i}\}+\alpha_i$ be the completion time of job $i\in\mathcal{J}$, if $i$ is the next to be executed from time $t$ and onward.
In addition, denote by $\mathcal{F}(t)=\{i:i\in\mathcal{J},C_i\leq t\}$ the set of completed jobs at time $t$ in $\mathcal{S}$.

\begin{algo}[ECTF]
\label{Algorithm:ECTF}
At each time $t$ that the machine becomes available, schedule the uncompleted job $\arg\min_{i\in\mathcal{J}\setminus\mathcal{F}(t)}\{\Gamma_i(t)\}$ with the earliest completion time. 
\end{algo}

\begin{theorem}
\label{Theorem:ECTF}
Algorithm~\ref{Algorithm:ECTF} achieves an approximation ratio $\rho\in[2,3+\frac{1}{\beta}]$.
\end{theorem}

\begin{proof}
We initially prove the upper bound. % here and show the lower bound in the Appendix.
Denote the non-interfering schedule and an optimal schedule by $\mathcal{S}$ and $\mathcal{S}^*$, respectively.
Number the jobs in increasing order $s_1<\ldots<s_n$ of their start times in $\mathcal{S}$.
That is, job $i\in\mathcal{J}$ is executed in the $i$-th position of $\mathcal{S}$.
Let $\pi(i)\in\{1,\ldots,n\}$ be the position at which job $i\in\mathcal{J}$ is executed in $\mathcal{S}^*$. 
Analogously, denote by $\gamma(i)\in\mathcal{J}$ the job executed in the $i$-th position of $\mathcal{S}$, for $i\in\{1,\ldots,n\}$.

We partition the set $\mathcal{J}$ of jobs into the subset $\mathcal{W}=\{i:i\geq\pi(i)\}$ of \emph{well-ordered jobs} whose position in $\mathcal{S}$ is greater than or equal to their position in $\mathcal{S}^*$ and the subset $\mathcal{I}=\{i:i<\pi(i)\}$ of \emph{inverted jobs} executed at a strictly smaller position in $\mathcal{S}$ compared to their position in $\mathcal{S}^*$.
% Figure~\ref{Figure:ECTF_Analysis} shows an inverted and a well-ordered job.
Consider an arbitrary inverted job $i\in\mathcal{I}$ executed in a subsequent position $\pi(i)\in\{i+1,\ldots,n\}$ in $\mathcal{S}^*$.
By a simple pigeonhole principle argument, a key observation is that there exists a job $j$ executed after $i$ in $\mathcal{S}$ and not later than the $i$-th position in $\mathcal{S}^*$, i.e.\ $\pi(j)\leq i<j$.
Clearly, job $j$ is well-ordered, i.e.\ $j\in\mathcal{W}$.

% \begin{figure}[!h]
% \begin{center}
% \input{images/ectf_analysis}
% \end{center}
% \caption{Illustration of an ECTF schedule $\mathcal{S}$ and an optimal schedule $\mathcal{S}^*$ using positions. Job $i$ is inverted and job $j$ is well-ordered. }
% \label{Figure:ECTF_Analysis}
% \end{figure}

Consider the start times $s_i$ and $s_i^*$ of job $i\in\mathcal{J}$ and the immediately preceding gaps $q_i$ and $q_i^*$ in $\mathcal{S}$ and $\mathcal{S}^*$, respectively. 
Based on the previous observation, define the set $\mathcal{K}_I=\{i:i\in\mathcal{I},\exists\; j\in\mathcal{W}\text{ s.t.\ }\pi(j)\leq i<j, r_j>s_i-q_i\}$ of \emph{critical inverted jobs}.
That is, for each job $k\in\mathcal{K}_I$, there exists a well-ordered job $\ell$ such that $\pi(\ell)\leq k<\ell$ and $\ell$ is released after $s_k-q_k$.
Given that ECTF executes $k$ before $\ell$, it must be the case that $\Gamma_k(s_k-q_k)\leq\Gamma_{\ell}(s_k-q_k)$, i.e.\ $C_k\leq (1+\beta)r_{\ell}+\alpha_{\ell}\leq C_{\ell}^*$. 
Thus, we get that $\sum_{i=1}^k(1+\beta)^{k-i+1}q_i+\sum_{i=1}^k(1+\beta)^{k-i}\alpha_i\leq\sum_{i=1}^{\pi(\ell)}(1+\beta)^{\pi(\ell)-i+1}q_{\gamma(i)}^*+\sum_{i=1}^{\pi(\ell)}(1+\beta)^{\pi(\ell)-i}\alpha_{\gamma(i)}$.
By taking into account that $\pi(\ell)\leq k$ and multiplying both sides with $(1+\beta)^{n-k}$:
\begin{align}
\label{Equation:ECTF_Critical}
\sum_{i=1}^k(1+\beta)^{n-i+1}q_i+\sum_{i=1}^k(1+\beta)^{n-i}\alpha_i\leq\sum_{i=1}^{k}(1+\beta)^{n-i+1}q_{\gamma(i)}^*+\sum_{i=1}^{k}(1+\beta)^{n-i}\alpha_{\gamma(i)}.
\end{align}

Next, define the set $\mathcal{K}_W=\{i:i\in\mathcal{W},q_i>0\}$ of \emph{critical well-ordered jobs}.
Consider a job $k\in\mathcal{K}_W$.
Given that $q_k>0$, job $k$ begins at its release time in $\mathcal{S}$, i.e.\ $s_k=r_k$.
That is, $C_k=(1+\beta)r_k+\alpha_k\leq C_{k}^*$, or equivalently $\sum_{i=1}^k(1+\beta)^{k-i+1}q_i+\sum_{i=1}^k(1+\beta)^{k-i}\alpha_i\leq\sum_{i=1}^{\pi(k)}(1+\beta)^{\pi(k)-i+1}q_{\gamma(i)}^*+\sum_{i=1}^{\pi(k)}(1+\beta)^{\pi(k)-i}\alpha_{\gamma(i)}$.
By taking into account the fact that $k$ is well-ordered, i.e.\ $\pi(k)\leq k$, and multiplying both sides of the inequality with $(1+\beta)^{n-k}$, we conclude that Eq.\ (\ref{Equation:ECTF_Critical}) holds for each job $k\in K_W$ as well.

Let $\mathcal{K}=\mathcal{K}_I\cup\mathcal{K}_W$ be the set of all critical jobs and consider the maximum index critical job $k=\max\{i:i\in\mathcal{K}\}$.
Next, denote by $\mathcal{W}_k=\{i:i>k,i\in\mathcal{W}\}$ and $\mathcal{I}_k=\{i:i>k,i\in\mathcal{I}\}$ the well-ordered and inverted jobs, respectively, of index strictly greater than the one of the maximum index critical job $k$.
For each job $i\in\mathcal{J}$ with $i>k$, either $i\in\mathcal{W}_k$, or $i\in\mathcal{I}_k$.
In the former case, since $i\in\mathcal{W}_k$, we have that $i\geq\pi(i)$.
Because $i>k$, it also holds that $q_i=0$. 
Therefore, $(1+\beta)^{n-i+1}q_i+(1+\beta)^{n-i}\alpha_i\leq(1+\beta)^{n-\pi(i)}\alpha_i$.
By summing over all jobs in $\mathcal{W}_k$, we get that
\begin{align}
\sum_{i\in\mathcal{W}_k}(1+\beta)^{n-i+1}q_i+\sum_{i\in\mathcal{W}_k}(1+\beta)^{n-i}\alpha_i \leq
\sum_{i\in\mathcal{W}_k}(1+\beta)^{n-\pi(i)}\alpha_i.
\end{align}
In the latter case, i.e.\ $i\in\mathcal{I}_k$, because $i<\pi(i)$, the pigeonhole principle argument mentioned earlier implies that there exists a well-ordered job $j\in\mathcal{W}$ such that $\pi(j)\leq i<j$.
Let $t=s_i-q_i$.
Because $i>k$, i.e.\ $i$ is non-critical, it must be the case that $r_j\leq t$.
Due to the ECTF policy and given that job $i$ is executed before job $j$ in $\mathcal{S}$, we have that
\begin{align*}
\Gamma_i(t)\leq\Gamma_j(t) & \Rightarrow (1+\beta)\max\{t,s_i\} + \alpha_i \leq (1+\beta)\max\{t,r_j\} + \alpha_j \\
& \Rightarrow (1+\beta)q_i + \alpha_i \leq \alpha_j.
\end{align*}
Taking also into account that $\pi(j)\leq i$, $(1+\beta)^{n-i+1}q_i+(1+\beta)^{n-i}\alpha_i\leq (1+\beta)^{n-\pi(j)}\alpha_j$.
We pick such a well-ordered job $j$ arbitrarily, match it with $i$, and denote it by $\mu(i)=j$. 
Let $\mathcal{M}_j=\{i:i\in\mathcal{I},\mu(i)=j\}$ be the set of inverted jobs matched with a job $j\in\mathcal{W}$.
If $i\in\mathcal{M}_j$, then it clearly holds that $\pi(j)\leq i$.
Thus, based on the weak pseudomatching bound (Lemma~7),
\begin{align*}
\sum_{i\in\mathcal{M}_j}[(1+\beta)^{n-i+1}q_i+(1+\beta)^{n-i}\alpha_i] &\leq\sum_{i=\pi(j)}^n(1+\beta)^{n-i}\alpha_j \\
& =\left[\frac{(1+\beta)^{n-\pi(j)+1}-1}{(1+\beta)-1}\right]\alpha_j \\
&\leq \left(1+\frac{1}{\beta}\right)(1+\beta)^{n-\pi(j)}\alpha_j
\end{align*}
That is, we get that:
\begin{align}
\sum_{i\in\mathcal{I}_k}[(1+\beta)^{n-i+1}q_i+(1+\beta)^{n-i}\alpha_i]
& = \sum_{j\in\mathcal{W}}\sum_{i\in\mathcal{M}_j}[(1+\beta)^{n-i+1}q_i+(1+\beta)^{n-i}\alpha_i] \nonumber \\
& \leq \left(1+\frac{1}{\beta}\right)\left[\sum_{j\in\mathcal{W}} (1+\beta)^{n-\pi(j)}\alpha_{j}\right]
\end{align}
The algorithm achieves makespan:
\begin{align}
T =\sum_{i=1}^k(1+\beta)^{n-i+1}[q_i+(1+\beta)^{n-i}\alpha_i] + \sum_{i\in\mathcal{W}_k}(1+\beta)^{n-i}\alpha_i \nonumber \\ 
+ \sum_{i\in\mathcal{I}_k}[(1+\beta)^{n-i+1}q_i+(1+\beta)^{n-i}\alpha_i] \label{Equation:ECTF_Upper_Bound}
\end{align}
For the optimal makespan, it clearly holds that:
\begin{align}
T^*\geq \max\left\{\sum_{i=1}^n(1+\beta)^{n-i+1}q_{\gamma(i)}^*+\sum_{i=1}^n(1+\beta)^{n-i}\alpha_{\gamma(i)},\sum_{i\in\mathcal{W}}(1+\beta)^{n-\pi(i)}\alpha_i\right\}
\label{Equation:ECTF_Lower_Bound}
\end{align}
By Eq.\ (\ref{Equation:ECTF_Critical})-(\ref{Equation:ECTF_Lower_Bound}), we conclude that $T\leq(3+\frac{1}{\beta})T^*$. 

\paragraph{Lower Bound}
% The main part of the manuscript proves the upper bound.
Next, we show the lower bound. %We refer the reader to Figure \ref{fig:ECTFLower}.
We consider an instance with $n=2k$ jobs: $k$ short jobs and $k$ long jobs.
The $j$-th long job has $r_j^L=0$ and $\alpha_j^L=(1+\beta)B$, for $j\in\{1,\ldots,k\}$.
The $j$-th short jobs has release time $r_j^S=\sum_{i=1}^j(1+\beta)^{i-1}B$ and and $\alpha_j^S=0$, for $j\in\{1,\ldots,k\}$.
Let $\mathcal{S}$ be a schedule produced by ECTF.
We show by induction that all small jobs are executed before all long jobs in $\mathcal{S}$, i.e.\ the $j$-th short job is executed during $[\sum_{i=1}^j(1+\beta)^{i-1}B, \sum_{i=1}^j(1+\beta)^iB)$. 
For $j=1$, the small job completes at $(1+\beta)B$ and any long job has completion time $\geq(1+\beta)B$ in any feasible schedule.
Next, assume that our claim is true for some $j\in\{1,\ldots,k-1\}$.
Since the $j$-th short job has completion time $C_j=\sum_{i=1}^j(1+\beta)^iB$, the $(j+1)$-th job begins at its release time $r_{j+1}>C_j$ and completes at $\sum_{i=1}^{j+1}(1+\beta)^iB$, while any long job would complete at $\geq C_{j+1}$ if it began at $C_j$.
Therefore,
\begin{align*}
T=\sum_{i=1}^{2k}(1+\beta)^iB=\frac{(1+\beta)B}{\beta}\left[(1+\beta)^{2k}-1\right].
\end{align*}
On the other hand, the optimal solution executes all long jobs before all short jobs and has makespan:
\begin{align*}
T^*=\sum_{i=k+1}^{2k}(1+\beta)^iB=\frac{(1+\beta)^{k+1}B}{\beta}\left[(1+\beta)^{k}-1\right].
\end{align*}
Therefore, we conclude that 
\begin{align*}
\beta T & = (1+\beta)^{2k+1}B-(1+\beta)B \\
& =\left[(1+\beta)^{2k+1}B-(1+\beta)^{k+1}B\right] + \left[(1+\beta)^{k+1}B-(1+\beta)B\right] \\
& \leq \left[1+\frac{1}{(1+\beta)^k}\right]\beta T^*
\end{align*}
Hence, $T/T^*\geq(1+\frac{1}{(1+\beta)^k})$.
If $\beta=\omega(1)$, then $\mathcal{S}$ is 2-approximate.
% If $\beta\geq\frac{1}{k}$, then $\mathcal{S}$ achieve an approximation ratio $\geq(1+\frac{1}{k})^k\geq(1+\frac{1}{e})$.
\end{proof}

\section{Total Completion Time Objective}
\label{Section:Sum_Completion_Times}

This section explores relationships between the problems of minimizing the makespan $\max_{i\in\mathcal{J}}\{C_i\}$ and the sum of completion times $\sum_{i\in\mathcal{J}}C_i$.
Theorem \ref{Theorem:Sum_Of_Completion_Times_To_Makespan} shows that any $O(1)$-approximate schedule for the former is also $O(1)$-approximate for the latter. 

\begin{theorem}
\label{Theorem:Sum_Of_Completion_Times_To_Makespan}
Any $\rho$-approximation algorithm for minimizing the sum of completion times is $(1+\rho)$-approximate for minimizing the makespan.
\end{theorem}
\begin{proof}
Suppose that $\mathcal{S}$ and $\mathcal{S}^*$ are a $\rho$-approximate schedule for minimizing the sum of completion times and an optimal schedule for minimizing makespan, while $T$ and $T^*$ are the makespans of the two schedules.
Assuming that $s_i$, $C_i$, and $q_i$ the start time, completion time, and gap associated with job $i\in\mathcal{J}$ in $\mathcal{S}$, we similarly define $s_i^*$, $C_i^*$, and $q_i^*$ for $\mathcal{S}^*$.
Given that $\mathcal{S}$ is $\rho$-approximate for minimizing $\sum_{i=1}^nC_i$, the job start times in the two schedules can be related as follows:
\begin{align*}
\sum_{i=1}^nC_i\leq\rho\left[\sum_{i=1}^nC_i^*\right]\Rightarrow 
& \sum_{i=1}^n[(1+\beta)s_i+\alpha_i] \leq \rho\left[\sum_{i=1}^n[(1+\beta)s_i^*+\alpha_i]\right] \\
\Rightarrow & \sum_{i=1}^ns_i \leq \rho\left[\sum_{i=1}^ns_i^*\right] + \frac{\rho-1}{1+\beta}\left[\sum_{i=1}^n\alpha_i\right]
\end{align*}
Observe that $\sum_{i=1}^nq_i\leq r_{\max}$, where $r_{\max}=\max_{i=1}^n\{r_i\}$ is the maximum release time, because gaps may only occur before release times in a canonical schedule.
To upper bound the makespan of $\mathcal{S}$:
\begin{align*}
T & = \sum_{i=1}^n[q_i+p_i(s_i)] \\ 
& = \sum_{i=1}^nq_i + \sum_{i=1}^n\alpha_i + \beta\left[\sum_{i=1}^ns_i\right] \\
& \leq r_{\max} + \frac{1+\rho\beta}{1+\beta}\left[\sum_{i=1}^n\alpha_i\right] + \rho\beta\left[\sum_{i=1}^ns_i^*\right] \\
& \leq r_{\max} + \rho\left[\sum_{i=1}^n(\beta s_i^*+\alpha_i)\right]
\end{align*}
The last inequality follows from the fact that $\rho\geq 1$.
Given that $T^*\geq\max\{\sum_{i=1}^n(\beta s_i^*+\alpha_i),r_{\max}\}$, we conclude that $T\leq(1+\rho)T^*$. 
\end{proof}

% This section extends known results of the time-dependent scheduling literature to the problem of minimizing the makespan with linearly deteriorating jobs that arrive over time.
% Scheduling jobs according to shortest fixed processing first computes an optimal schedule when all release times are equal \cite{Gupta1988}. 
% We show that a straightforward extension of the shortest fixed processing time first algorithm is $\Omega((1+\beta)^n)$-approximate when the release times are arbitrary.

Theorem~\ref{Theorem:Makespan_To_Sum_Of_Completion_Times} shows that any $O(1)$-approximate schedule for $\max_{i\in\mathcal{J}}\{C_i\}$ is $O(1+\frac{1}{\beta})$-approximate for $\sum_{i\in\mathcal{J}}C_i$.
This result and Theorem~\ref{Theorem:ECTF} directly imply that there exists an $O(1)$-approximation algorithm for minimizing the sum of completion times when $\beta=\Omega(1)$. 

\begin{theorem}
\label{Theorem:Makespan_To_Sum_Of_Completion_Times}
Any $\rho$-approximation algorithm for minimizing the makespan is $(1+\frac{1}{\beta})\rho$-approximate for minimizing the sum of completion times.
\end{theorem}
\begin{proof}
Suppose that $\mathcal{S}$ and $\mathcal{S}^*$ are a $\rho$-approximate schedule for minimizing the makespan and an optimal schedule for minimizing the sum of completion times, respectively. 
Assuming that $s_i$, $C_i$, and $q_i$ are the start time, completion time, and gap associated with job $i\in\mathcal{J}$ in $\mathcal{S}$, we similarly define $s_i^*$, $C_i^*$, and $q_i^*$ for $\mathcal{S}^*$.
Given that $\mathcal{S}$ is $\rho$-approximate for the makespan objective and the fact that $\sum_{i=1}^nq_i^*\leq r_{\max}$, we get that:
\begin{align*}
T \leq \rho T^* & \Rightarrow 
\sum_{i=1}^n[\beta s_i+\alpha_i] \leq \rho\left[\sum_{i=1}^nq_i^*+\sum_{i=1}^n[\beta s_i^*+\alpha_i]\right] \\
& \Rightarrow \sum_{i=1}^n s_i \leq \rho\left[\sum_{i=1}^ns_i^*\right] +
\frac{\rho}{\beta}\left[r_{\max}+\sum_{i=1}^n\alpha_i\right]
\end{align*}
Since $\rho\geq 1$, we can upper bound the sum of completion times of $\mathcal{S}$ as follows:
\begin{align*}
\sum_{i=1}^nC_i 
& = (1+\beta)\left[\sum_{i=1}^ns_i\right]+\sum_{i=1}^n\alpha_i \\
& \leq \rho\left[\sum_{i=1}^n[(1+\beta)s_i^*+\alpha_i]\right] + 
\frac{\rho}{\beta}\left[r_{\max}+\sum_{i=1}^n\alpha_i\right] 
\end{align*}
Because $\sum_{i=1}^nC_i^*=\sum_{i=1}^n[(1+\beta)s_i^*+\alpha_i]\geq r_{\max}+\sum_{i=1}^n\alpha_i$, we get that $\sum_{i=1}^nC_i\leq(1+\frac{1}{\beta})\rho\left[\sum_{i=1}^n C_i^*\right].$
\end{proof}

\begin{corollary}
Algorithm~\ref{Algorithm:ECTF} is $O(1+\frac{1}{\beta^2})$-approximate for minimizing $\sum_{i\in\mathcal{J}}C_i$.
\end{corollary}

\section{Concluding Remarks}
\label{Section:Conclusion}

We obtain new approximation results for time-dependent scheduling with uniformly deteriorating processing time functions $p_i(s_i)=\alpha_i+\beta\cdot s_i$, the makespan and the total completion time objectives.
The approximability of the more general problems with arbitrary linear deterioration remains an intriguing open question.
We expect our bounding framework based on pseudomatchings to be useful for follow-up work.
The key technical difficulty would be extending the proposed bounds to account for the different deteriorating rates.
We leave this as a future direction.

\bibliographystyle{elsarticle-harv} 
\bibliography{refs}

%% The Appendices part is started with the command \appendix;
%% appendix sections are then done as normal sections
% \appendix

% \section{Common Release Time}
% % \label{}

% \begin{theorem}
% Scheduling the jobs in decreasing order $\alpha_1\leq\ldots\leq\alpha_n$ of their sizes is optimal when the jobs have identical release times.
% \end{theorem}
% \begin{proof}
% Assume for contradiction that there exists an optimal schedule $\mathcal{S}^*$ that does not schedule the jobs w.r.t.\ to their non-decreasing order. 
% Then, the schedule contains an inversion $(i,j)$, i.e.\ a pair of consecutively executed jobs $i$ and $j$ s.t.\ $i$ is executed right before $j$, but 

% \end{proof}

\end{document}